\documentclass[11pt,a4paper]{article}
\usepackage[utf8]{inputenc}
\usepackage[T1]{fontenc}
\usepackage{amsmath,amssymb,amsthm,graphicx,float,enumitem,xcolor,subcaption}
\usepackage[margin=2.5cm]{geometry}
\usepackage[hidelinks]{hyperref}
\newtheorem{proposition}{Proposition}[section]

\renewcommand{\proof}{\par\noindent\textit{Proof.}\ }
\title{The influence of decay in dynamic models of gene regulatory networks}
\author{Oleksandr Pokutnyi$^{1,2}$ \and Felix Sadyrbaev$^{3,4}$\\[1ex]
\small $^1$ Kyiv-Mohyla Academy, Kyiv, Ukraine\\
\small $^2$ Institute of Mathematics, NAS of Ukraine, Kyiv, Ukraine\\
\small $^3$ Institute of Mathematics and Computer Science, University of Latvia, Riga, Latvia\\
\small $^4$ Daugavpils University, Daugavpils, Latvia}
\date{}
\begin{document}
\maketitle
\noindent\textbf{Keywords:} gene regulatory networks; dynamical systems; stability; characteristic polynomial; decay rates.
\bigskip

{\bf Abstract.} We investigate the influence of decay rates on the dynamics of three-dimensional gene regulatory networks modeled by nonlinear differential equations with logistic sigmoidal regulatory functions. Particular attention is paid to the existence, location, and stability of equilibrium points in systems involving numerous interacting parameters. An explicit analytical representation of a central equilibrium is obtained under suitable relations between the regulatory weights, thresholds, and decay coefficients, allowing its location to be determined without solving the corresponding nonlinear algebraic system numerically. General formulas for the characteristic polynomial of the linearized system are derived, and spectral conditions for instability are established for several classes of regulatory interactions. We demonstrate how variations in decay rates can change the number and stability of equilibria, transform oscillatory dynamics, and influence the emergence or suppression of complicated attractors. The possibility of chaotic behavior is investigated through eigenvalue configurations, Lyapunov exponents, and the Kaplan–Yorke dimension. Numerical examples illustrate the sensitivity of the dynamics to decay parameters, including the suppression of previously observed chaotic behavior. The results provide an analytical framework for investigating equilibrium configurations, stability transitions, and complex dynamics in parametrized gene regulatory networks.

\section{Introduction}
\label{sec:intro}

Mathematical biology has become a broad field for biologists and mathematicians looking to apply their theories. This is due to the effectiveness of mathematical methods, on the one hand, and the importance of the subject for natural sciences and practical purposes. The study of gene networks takes a special place.

Gene regulatory networks (GRNs) exist in any cell of any living organism. These networks
regulate vitally important functions, including morphogenesis and reactions to external factors. A gene regulatory network regulates gene expression — that is, when, where, and how strongly genes are turned on or off in a cell. A gene regulatory network is a set of genes, proteins, and regulatory interactions
that determine how gene expression changes over time. Direct investigation of GRNs is hard, as it is based on collecting experimental data and analyzing it. Collecting experimental data is a work-consuming and time-consuming process, and systematization, analysis, and interpretation of the obtained data are even more complicated. Mathematical modeling is an effective and promising way of making these processes easier. An effective mathematical model can be studied on its own and can give recommendations for researching real networks. One of the ways to model GRNs is by using dynamical systems. This means that gene expression levels are treated as state variables that evolve according to some rules. These rules are formulated in a language of mathematics as systems of ordinary differential equations (ODEs) of the form  
\begin{equation} \label{zero}
dx_i/dt
= f_i(x_1, x_2,.., x_n; P), 
\end{equation}
where $x_i$ stands for the expression level (mRNA/protein) of the $i$-th gene, and $f_i$ are regulation functions (they represent relations like activation, repression, degradation, noise, etc.), and $P$ is a set of parameters (binding strengths, thresholds, and degradation rates. The network structure (who regulates who) is encoded in the structure of functions $f_i.$
A gene regulatory network (GRN) can be modeled in various ways, and modeling using dynamical systems is only one of them. These models are deterministic, and they focus on average behavior, fixed points, oscillatory behaviors, and bifurcations. The nonlinearities $f_i$ are represented by sigmoidal functions, which there are many. In most problems, the Hill functions of the form $\frac{x^n}{x^n+\Theta^n}$ are applied. In the recent decade, the majority of studies use the logistic sigmoidal functions of the form $f(z)=\frac{1}{1+e^{-\mu (z-\theta)}}.$ These functions possess many convenient properties for mathematical study (\cite{Belgacem1}). These systems possess an invariant subset in a phase space. This subset traps trajectories, and further analysis can be reduced to this subset. Critical points always exist (in that subset) and can be multiple, but their number is limited. The critical points can be studied using the standard tools, such as linearization around them and considering the eigenvalues of the Jacobian $J.$ This local analysis can tell if the equilibrium state is stable, unstable, or oscillatory. One should consider bistability, even multistability, and switches. These mathematical phenomena may correspond to real processes that can be observed in genetic networks. One can meet limit cycles in dynamical systems occurring in models of GRNs. We will call such systems GRN-type systems. More attractors can exist, including more complicated periodic solutions and attractors, exhibiting chaotic behavior. The problem of finding and detecting attractors in related phase spaces is one of the most intriguing. Knowledge of attractors allows prediction of solution behavior over future time intervals. Basins of attraction are decisive for predicting future behavior of solutions. In a corresponding network, knowledge of attractors in a model can be used to provide recommendations on how to force a network to evolve in a needed direction. There are examples in the literature of using this methodology for biomedical purposes.  

The interested reader may try to become acquainted with recent advances in GRN modeling theory by reading the following articles: \cite{Alakwaa}, \cite{Polynikis}, \cite{Barbuti}, \cite{Aguirre}, \cite{Belova}, \cite{Meng}. Applications of models and various properties of GRNs were considered in \cite{Wang}, \cite{Cornelius}, \cite{Leukemia}, \cite{Ageing}, \cite{Ogorelova}. Chaotic behavior was studied in \cite{Das2000}, \cite{Das2002}, \cite{Mukherjee}, \cite{Kozlovska}. GRNs construction looking for certain properties was studied in \cite{Kozlovska2022}. Oscillatory properties of 3D models of GRNs were researched in \cite{Silvans}. Examples and samples of the work with three-dimensional (3D) systems, modeling GRNs, are presented in Section 2. This technical part shows how to construct a network
with given properties. Recent works by I. Belgacem  (\cite{Belgacem1} and more) provide multiple arguments in favor of logistic sigmoids.  Section 3 considers the general case of the 3D
system. Formulas for the computation of the characteristic polynomial for a critical point
are obtained. Their application allows consideration of several particular cases, revealing the influence of decay on solutions of the system.

\section{Main results}

In this section we consider the three dimensional GRN type system. We analyze critical points, establish non stability results for some particular cases. Since all critical points can exist only in the region $\{0<x_1<1/v_1, 0<x_2<1/v_2, 0<x_3<1/v_3\}$, and this region is trapping all incoming trajectories, the non-stability results for critical points are important for detecting attractors of more complicated structure, the stable closed trajectories included.
Consider the GRN-type system in the form
\begin{equation} \label{GRNsystem}
\left\{\begin{array}{rl}
x'_1 =  \frac{1}{1+e^{-\mu_1 (w_{11} x_1 + w_{12} x_2 + w_{13} x_3 - \theta_1)}}- v_1 x_1,    \\\\
x'_2 = \frac{1}{1+e^{-\mu_2 (w_{21} x_1 + w_{22} x_2 + w_{23} x_3 - \theta_2)}} - v_2 x_2, \\\\
x'_3 = \frac{1}{1+e^{-\mu_3 (w_{31} x_1 + w_{32} x_2 + w_{33} x_3 -
\theta_3)}}- v_3 x_3,
\end{array}\right.
\end{equation}
which can be rewritten as
$$
x' = - v_{1}x + \frac{1}{1 + e^{-\frac{\mu_{1}}{2}\left(\frac{w_{11}}{v_{1}}2v_{1}x + \frac{w_{12}}{v_{2}}2v_{2}y + \frac{w_{13}}{v_{3}}2v_{3}z - 2\theta_{1} \right)}},
$$
$$
y' = - v_{2}y + \frac{1}{1 + e^{-\frac{\mu_{2}}{2}\left(\frac{w_{21}}{v_{1}}2v_{1}x + \frac{w_{22}}{v_{2}}2v_{2}y + \frac{w_{23}}{v_{3}}2v_{3}z - 2\theta_{2} \right)}},
$$
$$
z' = -v_{3}z +  \frac{1}{1 + e^{-\frac{\mu_{3}}{2}\left(\frac{w_{31}}{v_{1}}2v_{1}x + \frac{w_{32}}{v_{2}}2v_{2}y + \frac{w_{33}}{v_{3}}2v_{3}z - 2\theta_{3} \right)}}.
$$

The system of equations for the critical points can be rewritten in the following form
$$
v_{1}x = \frac{1}{1 + e^{-\frac{\mu_{1}}{2}\left(\frac{w_{11}}{v_{1}}2v_{1}x + \frac{w_{12}}{v_{2}}2v_{2}y + \frac{w_{13}}{v_{3}}2v_{3}z - 2\theta_{1} \right)}},
$$
$$
v_{2}y = \frac{1}{1 + e^{-\frac{\mu_{2}}{2}\left(\frac{w_{21}}{v_{1}}2v_{1}x + \frac{w_{22}}{v_{2}}2v_{2}y + \frac{w_{23}}{v_{3}}2v_{3}z - 2\theta_{2} \right)}},
$$
$$
v_{3}z = \frac{1}{1 + e^{-\frac{\mu_{3}}{2}\left(\frac{w_{31}}{v_{1}}2v_{1}x + \frac{w_{32}}{v_{2}}2v_{2}y + \frac{w_{33}}{v_{3}}2v_{3}z - 2\theta_{3} \right)}}
$$
If we have that $v_{1}x = 0.5$, $v_{2}y = 0.5, v_{3}z = 0.5$ then we obtain the following conditions on  parameters $\theta_{i}$: 
\begin{equation} \label{thetai}
\frac{w_{11}}{v_{1}} + \frac{w_{12}}{v_{2}} + \frac{w_{13}}{v_{3}} = 2\theta_{1},
$$
$$
\frac{w_{21}}{v_{1}} + \frac{w_{22}}{v_{2}} + \frac{w_{23}}{v_{3}} = 2\theta_{2}
\end{equation}
$$
\frac{w_{31}}{v_{1}} + \frac{w_{32}}{v_{2}} + \frac{w_{33}}{v_{3}} = 2\theta_{3}.
$$
In such a way $x = \frac{0.5}{v_{1}}, y = \frac{0.5}{v_{2}}, z = \frac{0.5}{v_{3}}$ is a solution of our system. 
If, for example $v_{1} = v_{2} = v_{3} = v$ then $x = y = z = \frac{0.5}{v}$ and the solvability conditions  take form
$$
w_{11} + w_{21} + w_{31} = 2\theta_{1}v
$$
$$
w_{21} + w_{22} + w_{32} = 2\theta_{2}v
$$
$$
w_{31} + w_{32} + w_{33} = 2\theta_{3}v.
$$
In the case when $w_{11} = 1, w_{12} = 0, w_{13} = -1, w_{21} = -1, w_{22} = 1, w_{23} = 0, w_{31} = 0, w_{32} = -1, w_{33} = 1 $ we have that $\theta_{1} + \theta_{2} + \theta_{3} = 0$, 
$$
v_{3} = \frac{v_{1}}{1 - 2\theta_{1}v_{2}}, \quad v_{2} = \frac{v_{1}}{1 + 2\theta_{2}v_{1}},
$$
and corresponding solutions we can represent in the following form
$$
x = \frac{0.5}{v_{1}}, \quad y = \frac{0.5(1 + 2\theta_{2}v_{1})}{v_{1}}, \quad z = \frac{0.5(1 - 2\theta_{1}v_{1})}{v_{1}} 
$$
For example $x = 0.5, y = 0.6, z = 0.3, \theta_{1} = 0.2, \theta_{2} = 0.1, \theta_{3} = -0.3,$
for any $v_{1}, v_{2}, v_{3}, \mu_{1}, \mu_{2}, \mu_{3}$. In this case we have that
$$
P_{1}(x, y, z) = \frac{1}{1 + e^{-\mu_{1}(x - z - 0.2)}} - v_{1}x,
$$
$$
P_{2}(x, y, z) = \frac{1}{1 + e^{-\mu_{2}( -x + y - 0.1)}} - v_{2}y,
$$
$$
P_{3}(x, y, z) = \frac{1}{1 + e^{-\mu_{3}(-y + z + 0.3)}} - v_{3}z.
$$
Then 
$$
P_{1x}(x_{0}, y_{0}, z_{0}) = \mu_{1}w_{11}(v_{1}x_{0} - v_{1}^{2}x_{0}^{2}) - v_{1}, P_{1y}'(x_{0}, y_{0}, z_{0}) = \mu_{1}w_{12}v_{1}x_{0}(1 - v_{1}x_{0}),
-$$
$$
P_{1z}'(x_{0}, y_{0}, z_{0}) = \mu_{1}w_{13}(v_{1}x_{0} - v_{1}^{2}x_{0}^{2}),
$$
$$
P_{2x}'(x_{0}, y_{0}, z_{0}) = \mu_{2}w_{21}(v_{2}y_{0} - v_{2}^{2}y_{0}^{2}), ~~ P_{2y}'(x_{0}, y_{0}, z_{0}) = \mu_{2}w_{22}(v_{2}y_{0} - v_{2}^{2}y_{0}^{2}) - v_{2}, ~~ 
$$
$$
P_{2z}'(x_{0}, y_{0}, z_{0}) = \mu_{3}w_{23}(v_{2}y_{0} - v_{2}^{2}y_{0}^{2}),
$$
$$
P_{3x}'(x_{0}, y_{0}, z_{0}) = \mu_{3}w_{31}v_{3}z_{0}(1 - v_{3}z_{0}), P_{3y}'(x_{0}, y_{0}, z_{0}) = \mu_{3}w_{32}(v_{3}z_{0} - v_{3}^{2}z_{0}^{2}), 
$$
$$
P_{3z}'(x_{0}, y_{0}, z_{0}) =  \mu_{3}w_{33}(v_{3}z_{0} - v_{3}^{2}z_{0}^{2}) - v_{3}.
$$
Linearized system has the following form
$$
u_{1}' = (0.25 \mu_{1}w_{11} - v_{1})u_{1} + 0.25\mu_{1}w_{12}u_{2} + 0.25 \mu_{1}w_{13}u_{3},
$$
$$
u_{2}' = 0.25w_{21}\mu_{2}u_{1} + (0.25\mu_{2}w_{22} - v_{2})u_{2} + 0.25\mu_{3}w_{23}u_{3},
$$
$$
u_{3}' = 0.25\mu_{3}w_{31}u_{1} + 0.25\mu_{3}w_{32}u_{2} + (0.25\mu_{3}w_{33} - v_{3})u_{3}.
$$
Consider for the simplicity 
$$
\mu_{1} = \mu_{2} = \mu_{3} = 1, v_{1} = 1, v_{2} = \frac{5}{6}, v_{3} = \frac{5}{3},
$$
corresponding weights have the form
 $w_{11} = 1, w_{12} = 0, w_{13} = -1, w_{21} = -1, w_{22} = 1, w_{23} = 0, w_{31} = 0, w_{32} = -1, w_{33} = 1 $. Then the characteristic determinant has the form
$$
\left| \begin{array}{ccc} 0.25 - 1 - \lambda & 0 & -0.25  \\ 
-0.25 & -\frac{7}{12} - \lambda  & 0 \\
0 & -0.25 & -\frac{17}{12} - \lambda
\end{array}\right| = 0.
$$
In this case we obtain the following characteristic equation 
$$
p(\lambda) = -(\frac{3}{4} + \lambda)(\lambda + \frac{17}{12})(\lambda - \frac{19}{24}) - \frac{1}{64} = 0
$$
or
$$
\lambda^3 + \frac{11}{8}\lambda^2 - \frac{47}{72}\lambda - \frac{317}{384} = 0
$$
Using the Hurwitz criterion for the matrix
$$
\left( \begin{array}{ccccc} \frac{11}{8} & -\frac{317}{384} & 0 \\
1 & -\frac{47}{72} & 0 \\
0 & \frac{11}{8} & -\frac{317}{384}
\end{array}\right)
$$
we can obtain
$$
\delta_{1} = \frac{11}{8} > 0,
$$
$$
\delta_{2} < 0
$$
and corresponding instability ($\lambda_{1} > 0, \lambda_{2} \neq \lambda_{3} < 0$).

In the general case we obtain the following system (corresponding linearization) $\mu_{i} = 4$:
$$
u_{1}' = (w_{11} - v_{1})u_{1} + w_{12}u_{2} + w_{13}u_{3},
$$
$$
u_{2}' = w_{21} u_{1} + (w_{22} - v_{2})u_{2} + w_{23}u_{3},
$$
$$
u_{3}' = w_{31}u_{1} + w_{32}u_{2} + (w_{33} - v_{3})u_{3}.
$$
with characteristic equation
$$
-(\lambda - w_{11} + v_{1})(\lambda - w_{22} + v_{2})(\lambda - w_{33} + v_{3}) + w_{12}w_{23}w_{31} + w_{13}w_{21}w_{32} - 
$$
$$
- w_{13}w_{31}(w_{22} - v_{2} - \lambda) - w_{12}w_{21}(w_{33} - v_{3} - \lambda) - w_{23}w_{32}(w_{11} - v_{1} - \lambda) = 0
$$
Finally, we have that
$$
- \lambda^{3} + A \lambda^{2} + B \lambda + C = 0,
$$
where
$$
A = w_{11} + w_{22} + w_{33} - v_{1} - v_{2} - v_{3},
$$
$$
B = -\left((v_{1} - w_{11})(v_{3} - w_{33}) + (v_{2} - w_{22})(v_{3} - w_{33}) + (v_{1} - w_{11})(v_{2} - w_{22}) \right) +
$$
$$
+ w_{13}w_{31} + w_{12}w_{21} + w_{23}w_{32}, 
$$
$$
C=  - (v_{1} - w_{11})(v_{2} - w_{22})(v_{3} - w_{33}) +
$$
$$
+ w_{12}w_{23}w_{31} + w_{13}w_{21}w_{32} - w_{13}w_{31}w_{22} - w_{12}w_{21}w_{33} - w_{23}w_{32}w_{11}
$$
or in the form
\begin{align*}
&- \lambda^{3}
+\bigl(w_{11}+w_{22}+w_{33}-v_{1}-v_{2}-v_{3}\bigr)\lambda^{2}\\[4pt]
&\quad +\bigl(
- v_{1}v_{2} - v_{1}v_{3} + v_{1}w_{22} + v_{1}w_{33}
- v_{2}v_{3} + v_{2}w_{11} + v_{2}w_{33}\\[-2pt]
&\qquad\qquad\quad
+ v_{3}w_{11} + v_{3}w_{22}
- w_{11}w_{22} - w_{11}w_{33} + w_{12}w_{21} + w_{13}w_{31}
- w_{22}w_{33} + w_{23}w_{32}
\bigr)\lambda\\[6pt]
&\quad +\bigl(
- v_{1}v_{2}v_{3} + v_{1}v_{2}w_{33} + v_{1}v_{3}w_{22}
- v_{1}w_{22}w_{33} + v_{1}w_{23}w_{32}\\[-2pt]
&\qquad\qquad\quad
+ v_{2}v_{3}w_{11} - v_{2}w_{11}w_{33} + v_{2}w_{13}w_{31}
- v_{3}w_{11}w_{22} + v_{3}w_{12}w_{21}\\[-2pt]
&\qquad\qquad\quad
+ w_{11}w_{22}w_{33} - w_{11}w_{23}w_{32} - w_{12}w_{21}w_{33}
+ w_{12}w_{23}w_{31} + w_{13}w_{21}w_{32} - w_{13}w_{22}w_{31}
\bigr)
\;=\;0.
\end{align*}

In the case when $w_{11} = w_{22} = w_{33} = 1, w_{13} = w_{21} = w_{32} = -1, w_{12} = 0, w_{23} = 0, w_{31} = 0$ we have that 
\begin{equation}
\begin{aligned}
-\lambda^3 &+ \bigl(3-(v_1+v_2+v_3)\bigr)\lambda^2\\
&+\bigl(-v_1v_2-v_1v_3-v_2v_3+2(v_1+v_2+v_3)-3\bigr)\lambda\\
&-v_1v_2v_3+v_1v_2+v_1v_3+v_2v_3-(v_1+v_2+v_3)=0.
\end{aligned}
\end{equation}

In the case when $w_{11} = w_{22} = w_{33} = 0, w_{12} = w_{13} = w_{21} = w_{23} = w_{31} = w_{32} = 1$ we have that
\begin{equation}
-\lambda^3 - (v_1 + v_2 + v_3) \, \lambda^2 + \bigl(3 - (v_1 v_2 + v_1 v_3 + v_2 v_3)\bigr) \, \lambda + \bigl(2 + v_1+v_2+v_3-v_1 v_2 v_3\bigr) = 0.
\end{equation}

In the case when $w_{11} = w_{22} = w_{33} = 0, w_{12} = w_{13} = w_{21} = w_{23} = w_{31} = w_{32} = -1$ 
we have that \begin{equation}
\lambda^3
+ (v_1+v_2+v_3)\,\lambda^2
+ (v_1 v_2 + v_1 v_3 + v_2 v_3 - 3)\,\lambda
+ (v_1 v_2 v_3 - (v_1+v_2+v_3) + 2)
= 0
\end{equation}

\section*{Conditions for instability}

Let 
\[
S = v_1+v_2+v_3, \quad P = v_1v_2 + v_1v_3 + v_2v_3, \quad Q = v_1v_2v_3
\]
for convenience. For each of the three characteristic polynomials, the conditions for having one positive real root and two complex conjugate roots with negative real parts are:

\subsection*{Case 1: $w_{11} = w_{22} = w_{33} = 1$, $w_{13} = w_{21} = w_{32} = -1$, $w_{12} = w_{23} = w_{31} = 0$}
The polynomial:
\[
-\lambda^3 + (3-S)\lambda^2 + (-P + 2S - 3)\lambda + (-Q + P - S) = 0
\]
Standard form (multiply by $-1$):
\[
\lambda^3 + (S-3)\lambda^2 + (P-2S+3)\lambda + (Q-P+S) = 0
\]

Conditions:
$$ D_1 = 18 a_2 a_1 a_0 - 4 a_2^3 a_0 + a_2^2 a_1^2 - 4 a_1^3 - 27 a_0^2 < 0 ~(\text{two complex roots}),
$$
$$
\lambda_1 > 3-S (\text{real part of complex roots negative}),
a_0 = Q - P + S < 0  (\text{product condition})
$$
where $a_2 = S-3$, $a_1 = P-2S+3$, $a_0 = Q-P+S$.

\subsection*{Case 2: $w_{11} = w_{22} = w_{33} = 0$, $w_{12} = w_{13} = w_{21} = w_{23} = w_{31} = w_{32} = 1$}
The polynomial:
\[
-\lambda^3 - S \lambda^2 + (3-P)\lambda + (2+S-Q) = 0
\]
Standard form (multiply by $-1$):
\[
\lambda^3 + S \lambda^2 + (P-3)\lambda + (Q-S-2) = 0
\]

Conditions:
\begin{equation}
\begin{cases}
D_2 = 18 a_2 a_1 a_0 - 4 a_2^3 a_0 + a_2^2 a_1^2 - 4 a_1^3 - 27 a_0^2 < 0,\\[1mm]
\lambda_1 >-a_2 = -S,\\[1mm]
a_0 = Q-S-2 < 0
\end{cases}
\end{equation}
where $a_2 = S$, $a_1 = P-3$, $a_0 = Q-S-2$.

\subsection*{Case 3: $w_{11} = w_{22} = w_{33} = 0$, $w_{12} = w_{13} = w_{21} = w_{23} = w_{31} = w_{32} = -1$}
The polynomial:
\[
\lambda^3 + S \lambda^2 + (P-3)\lambda + (Q-S+2) = 0
\]

Conditions:
\begin{equation}
\begin{cases}
D_3 = 18 a_2 a_1 a_0 - 4 a_2^3 a_0 + a_2^2 a_1^2 - 4 a_1^3 - 27 a_0^2 < 0,\\[1mm]
\lambda_1 > -a_2 = -S,\\[1mm]
a_0 = Q-S+2 < 0
\end{cases}
\end{equation}
where $a_2 = S$, $a_1 = P-3$, $a_0 = Q-S+2$.

Let 
\[
S = v_1+v_2+v_3, \quad P = v_1v_2 + v_1v_3 + v_2v_3, \quad Q = v_1v_2v_3.
\]

For a cubic polynomial in standard form
\[
\lambda^3 + a_2 \lambda^2 + a_1 \lambda + a_0 = 0,
\]
the positive real root $\lambda_1$ (if it exists and the other two roots are complex with negative real part) can be computed using the Cardano formula. Define the depressed cubic
\[
\lambda = y - \frac{a_2}{3}, \quad y^3 + p y + q = 0,
\]
with
\[
p = a_1 - \frac{a_2^2}{3}, \quad q = \frac{2 a_2^3}{27} - \frac{a_2 a_1}{3} + a_0.
\]

Then
\[
\lambda_1 = \sqrt[3]{-\frac{q}{2} + \sqrt{\frac{q^2}{4} + \frac{p^3}{27}}} 
          + \sqrt[3]{-\frac{q}{2} - \sqrt{\frac{q^2}{4} + \frac{p^3}{27}}} 
          - \frac{a_2}{3}.
\]

\subsection*{Case 1: $w_{11}=w_{22}=w_{33}=1$, $w_{13}=w_{21}=w_{32}=-1$, $w_{12}=w_{23}=w_{31}=0$}
\[
a_2 = S-3, \quad a_1 = P-2S+3, \quad a_0 = Q-P+S,
\]
\[
p_1 = a_1 - \frac{a_2^2}{3}, \quad q_1 = \frac{2 a_2^3}{27} - \frac{a_2 a_1}{3} + a_0,
\]
\[
\lambda_1 = \sqrt[3]{-\frac{q_1}{2} + \sqrt{\frac{q_1^2}{4} + \frac{p_1^3}{27}}} 
          + \sqrt[3]{-\frac{q_1}{2} - \sqrt{\frac{q_1^2}{4} + \frac{p_1^3}{27}}} 
          - \frac{a_2}{3}.
\]

\subsection*{Case 2: $w_{11}=w_{22}=w_{33}=0$, $w_{12}=w_{13}=w_{21}=w_{23}=w_{31}=w_{32}=1$}
\[
a_2 = S, \quad a_1 = P-3, \quad a_0 = Q-S-2,
\]
\[
p_2 = a_1 - \frac{a_2^2}{3}, \quad q_2 = \frac{2 a_2^3}{27} - \frac{a_2 a_1}{3} + a_0,
\]
\[
\lambda_1 = \sqrt[3]{-\frac{q_2}{2} + \sqrt{\frac{q_2^2}{4} + \frac{p_2^3}{27}}} 
          + \sqrt[3]{-\frac{q_2}{2} - \sqrt{\frac{q_2^2}{4} + \frac{p_2^3}{27}}} 
          - \frac{a_2}{3}.
\]

\subsection*{Case 3: $w_{11}=w_{22}=w_{33}=0$, $w_{12}=w_{13}=w_{21}=w_{23}=w_{31}=w_{32}=-1$}
\[
a_2 = S, \quad a_1 = P-3, \quad a_0 = Q-S+2,
\]
\[
p_3 = a_1 - \frac{a_2^2}{3}, \quad q_3 = \frac{2 a_2^3}{27} - \frac{a_2 a_1}{3} + a_0,
\]
\[
\lambda_1 = \sqrt[3]{-\frac{q_3}{2} + \sqrt{\frac{q_3^2}{4} + \frac{p_3^3}{27}}} 
          + \sqrt[3]{-\frac{q_3}{2} - \sqrt{\frac{q_3^2}{4} + \frac{p_3^3}{27}}} 
          - \frac{a_2}{3}.
\]

\section*{Examples of characteristic polynomials and roots}

Here we provide explicit examples of parameters $v_1, v_2, v_3$ for each case, illustrating both complex-conjugate and real eigenvalue configurations, and the instability of the corresponding equilibria.

We set $\mu_i=4,$ $\theta_i=0.5\left(w_{i1}/v_1+w_{i2}/v_2+w_{i3}/v_3\right)$ for $i=1,2,3.$

\subsection*{Case 1: $w_{11}=w_{22}=w_{33}=1$, $w_{13}=w_{21}=w_{32}=-1$, $w_{12}=w_{23}=w_{31}=0$}

\textbf{Example 1:} $v_1=v_2=v_3=0.5$

\[
\lambda^3 - 1.5 \lambda^2 + 0.75 \lambda + 0.875 = 0
\]

 \[
\text{Roots: } \lambda_1 =-0.5, \quad \lambda_{2,3} \approx 1 \pm 0.866\, i
\]

The central critical point  $(1,1,1)$ is not attractive. Instead the stable limit cycle is an attractor. A set of twenty solutions tending to the limit cycle is shown in Fig. \ref{fig:one}.
\begin{figure}[h!] 
    \centering    \includegraphics[width=0.4\textwidth]{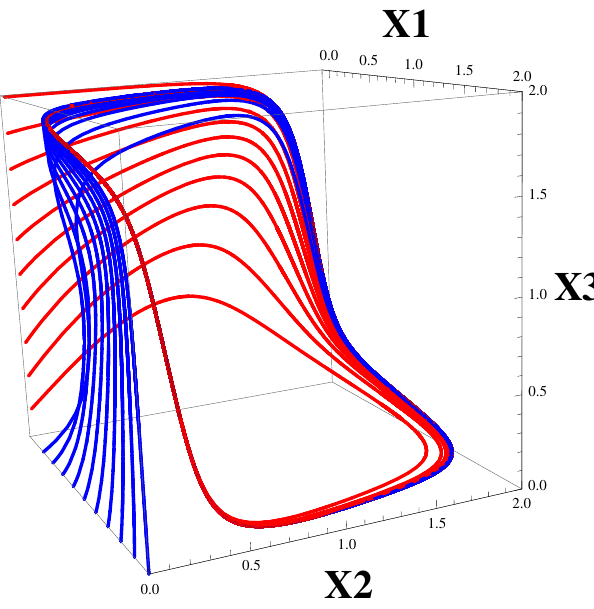}
    \caption{The phase space for Example 1.}
    \label{fig:one}
\end{figure}

\vskip0.3cm
\noindent \textbf{Example 2:} $v_1=0.5, v_2=1, v_3=0.8$
\[
\det(\lambda I-J)=(\lambda-0.5)\lambda(\lambda-0.2)+1
=\lambda^3-0.7\lambda^2+0.1\lambda+1=0.
\]
\[
\text{Roots: }\lambda_1\approx-0.787097,\qquad
\lambda_{2,3}\approx0.743549\pm0.847128\,i.
\]

\subsection*{Case 2: $w_{11}=w_{22}=w_{33}=0$, $w_{12}=w_{13}=w_{21}=w_{23}=w_{31}=w_{32}=1$}

\[
v_1=0.5, v_2=1, v_3=1
\]

\[
\lambda^3 + 2.5\lambda^2 - \lambda -4 =0
\]

\[
\text{Roots: } \lambda_1 \approx 1.186141, \; \lambda_{2} \approx -1.686141, \; \lambda_{3} = -2
\]

---

\subsection*{Case 3: $w_{11}=w_{22}=w_{33}=0$, $w_{12}=w_{13}=w_{21}=w_{23}=w_{31}=w_{32}=-1$}

\[
v_1=1, v_2=0.5, v_3=0.5
\]

\[
\lambda^3 + 2 \lambda^2 -1.75 \lambda +0.25 =0
\]

\[
\text{Roots:} \lambda_1 = 0.5, \quad \lambda_2 \approx 0.186141, \quad \lambda_3 \approx -2.686141
\]

\subsection*{Summary Table}
\vspace{0.5em}

\begin{table}[H]
\centering
\caption{Summary Table}
\label{Table 1}

\begin{tabular}{|c|c|c|c|}
\hline
Case & $v_1,v_2,v_3$ & Polynomial & Roots \\
\hline
1 & 0.5,0.5,0.5 & $\lambda^3 -1.5 \lambda^2 +0.75\lambda +0.875=0$ & $-0.5, 1\pm 0.866 i$ \\
1 & 0.5,1,0.8 & $\lambda^3 -0.7\lambda^2 +0.1\lambda +1=0$ & $-0.787,\;0.744\pm0.847i$ \\
2 & 0.5,1,1 & $\lambda^3 +2.5\lambda^2 -\lambda -4=0$ & $1.186,\;-1.686,\;-2$ \\
3 & 1,0.5,0.5 & $\lambda^3 +2\lambda^2 -1.75\lambda +0.25=0$ & $0.5,\;0.186,\;-2.686$ \\
\hline
\end{tabular}

\vspace{0.5em}
{\normalsize Examples of characteristic polynomials and roots for three cases.}
\end{table}

\section{General system}

Consider the general form of a 3D GRN-type system

\begin{equation} \label{main6}
\left\{\begin{array}{rl}
x'_1 =  \frac{1}{1+e^{-\mu_1 (w_{11} x_1 + w_{12} x_2 + w_{13} x_3 - \theta_1)}}- v_1 x_1,    \\\\
x'_2 = \frac{1}{1+e^{-\mu_2 (w_{21} x_1 + w_{22} x_2 + w_{23} x_3 - \theta_2)}} - v_2 x_2, \\\\
x'_3 = \frac{1}{1+e^{-\mu_3 (w_{31} x_1 + w_{32} x_2 + w_{33} x_3 -
\theta_3)}}- v_3 x_3.
\end{array}\right.
\end{equation}

This system can have attractors of any kind. For the equilibria (critical points) the standard local analysis can be performed.

Critical points are to be obtained from the relations
\begin{equation} \label{maincrpoint}
\left\{\begin{array}{rl}
v_1 x_1 =  \frac{1}{1+e^{-\mu_1 (w_{11} x_1 + w_{12} x_2 + w_{13} x_3 - \theta_1)}},    \\\\
v_2 x_2 = \frac{1}{1+e^{-\mu_2 (w_{21} x_1 + w_{22} x_2 + w_{23} x_3 - \theta_2)}}, \\\\
v_3 x_3 = \frac{1}{1+e^{-\mu_3 (w_{31} x + w_{32} x_2 + w_{33} x_3 -
\theta_3)}}.
\end{array}\right.
\end{equation}
\begin{proposition}  \label{prop1}
The point $(0.5/v_1, 0.5/v_2, 0.5/v_3 )$ is central in the set $Q_3=\{0<x_1<1/v_1, \: 0<x_2<1/v_2, \: 0<x_3<1/v_3\}.$ It is critical point for the system (\ref{main6}) if
\begin{itemize} [label=\textbullet]
\item $\theta_1=w_{11} 0.5/v_1+ w_{12} 0.5/v_2 + w_{13} 0.5/v_3$
\item $\theta_2=w_{21} 0.5/v_1+ w_{22} 0.5/v_2 + w_{23} 0.5/v_3$
\item $\theta_3=w_{31} 0.5/v_1+ w_{32} 0.5/v_2 + w_{33} 0.5/v_3$
\end{itemize}
\end{proposition}

\proof By direct checking (\ref{maincrpoint}).

\vskip0.3cm
There is at least one critical point, but there may be multiple, and all of them are in the set $Q_3.$

Critical points are subject to analysis. The~standard analysis of a critical point $x^*=(x_1^*,\ldots, x_n^*)$ involves the following~procedure. 

Linearization at a critical point $(x^*_1, \ldots, x^*_n)$ gives
\begin{equation} \label{mainlinearization}
\left\{\begin{array}{lrrl}
      \displaystyle u_1' = - v_1 u_1 + \mu_1 w_{11} g_1 u_1 + \mu_1 w_{12} g_1 u_2 +\mu_1 w_{13} g_1 u_3,\\
      \displaystyle u_2' = - v_2 u_2 + \mu_2 w_{21} g_2 u_1 + \mu_2 w_{22} g_2 u_2 +\mu_2 w_{23} g_2 u_3,\\
            \displaystyle u_3' = - v_3 u_3 + \mu_3 w_{31} g_3 u_1 + \mu_3 w_{32} g_3 u_2 + \mu_3 w_{33} g_3 u_3,\\
    \end{array} \right.
\end{equation}
where
\begin{equation} \label{maing1}
    g_1 = \frac{ e^{-\mu_1(w_{11} x^*_1 + w_{12} x^*_2 +  w_{13} x^*_3 - \theta_1)}}{[1 + e^{-\mu_1(w_{11} x^*_1 + w_{12} x^*_2 + w_{13} x^*_3 -
    \theta_1)}]^2},
\end{equation}
\begin{equation} \label{maing2}
    g_2 = \frac{ e^{-\mu_2(w_{21} x^*_1 + w_{22} x^*_2 +w_{23} x^*_3 - \theta_2)}}{[1 + e^{-\mu_2(w_{21} x^*_1 + w_{22} x^*_2 +  w_{23} x^*_3 -
    \theta_2)}]^2},
\end{equation}
\begin{equation} \label{maing3}
    g_3 = \frac{ e^{-\mu_3(w_{31} x^*_1 + w_{32} x^*_2 + w_{33} x^*_3 - \theta_3)}}{[1 + e^{-\mu_3(w_{31} x^*_1 + w_{32} x^*_2 + w_{33} x^*_3 - \theta_3)}]^2}.
\end{equation}
It follows from (\ref{maincrpoint}) that
\begin{equation} \label{prosto}
\left\{
\begin{array}{rl}
\displaystyle e^{-\mu_1(w_{11} x^*_1 + w_{12} x^*_2 +  w_{13} x^*_3 - \theta_1} =  \frac{1}{v_1 x^*_1}-1, \\\\
\displaystyle e^{-\mu_2(w_{21} x^*_1 + w_{22} x^*_2 +w_{23} x^*_3 - \theta_2} = \frac{1}{v_2 x^*_2}-1, \\\\
\displaystyle e^{-\mu_3(w_{31} x^*_1 + w_{32} x^*_2 + w_{33} x^*_3 - \theta_3} = \frac{1}{v_3 x^*_3}-1.
\end{array}
\right.
\end{equation}
Then $g_i=(1-v_i x^*_i)v_i x^*_i$ for $i=1,2,3,$ and the coefficient matrix $A$ of the linear system (\ref{mainlinearization}) can be written in the form
\begin{equation} \label{A}
A=\left(
\begin{matrix}
\mu_{1} w_{11} (1-v_1 x_{1}^{\ast}) v_1 x_{1}^{\ast} - v_1 
    & \mu_{1} w_{12} (1-v_1 x_{1}^{\ast}) v_1 x_{1}^{\ast} 
    & \mu_{1} w_{13} (1-v_1 x_{1}^{\ast}) v_1 x_{1}^{\ast} \\[6pt]
\mu_{2} w_{21} (1-v_2 x_{2}^{\ast}) v_2 x_{2}^{\ast} 
    & \mu_{2} w_{22} (1-v_2 x_{2}^{\ast}) v_2 x_{2}^{\ast} - v_2 
    & \mu_{2} w_{23} (1-v_2 x_{2}^{\ast}) v_2 x_{2}^{\ast} \\[6pt]
\mu_{3} w_{31} (1-v_3 x_{3}^{\ast}) v_3 x_{3}^{\ast} 
    & \mu_{3} w_{32} (1-v_3 x_{3}^{\ast}) v_3 x_{3}^{\ast} 
    & \mu_{3} w_{33} (1-v_3 x_{3}^{\ast}) v_3 x_{3}^{\ast} - v_3
\end{matrix}
\right).
\end{equation}
The characteristic equation $\det{\left(A-\lambda E\right)}=0,$ where $E$ is the unity matrix, is then a cubic equation with coefficients expressed explicitly in terms of the critical point.

The characteristic polynomial in the form
\begin{equation}
    P(\lambda)= c_3 \lambda^3+c_2 \lambda^2 + c_1 \lambda + c_0
\end{equation}
is ready. The characteristic numbers for a critical point $(x_{1}^{\ast},x_{2}^{\ast},x_{3}^{\ast})$ are solutions of the equation $P(\lambda)=0.$ We can apply it to several particular cases.

We know that $c_3=-1,$
\[
c_2=-v_1 - v_2 - v_3
+ v_1 w_{11} x_1 \mu_1
- v_1^2 w_{11} x_1^2 \mu_1
+ v_2 w_{22} x_2 \mu_2
- v_2^2 w_{22} x_2^2 \mu_2
+ v_3 w_{33} x_3 \mu_3
- v_3^2 w_{33} x_3^2 \mu_3
\]

\[
\begin{aligned}
c_1=&-v_1 v_2 - v_1 v_3 - v_2 v_3 
+ v_1 v_2 w_{11} x_1 \mu_1 
+ v_1 v_3 w_{11} x_1 \mu_1 
- v_1^2 v_2 w_{11} x_1^2 \mu_1 
- v_1^2 v_3 w_{11} x_1^2 \mu_1 \\
&\quad
+ v_1 v_2 w_{22} x_2 \mu_2 
+ v_2 v_3 w_{22} x_2 \mu_2 
- v_1 v_2^2 w_{22} x_2^2 \mu_2
- v_2^2 v_3 w_{22} x_2^2 \mu_2 \\
&\quad
+ v_1 v_2 w_{12} w_{21} x_1 x_2 \mu_1 \mu_2
- v_1 v_2 w_{11} w_{22} x_1 x_2 \mu_1 \mu_2 \\
&\quad
- v_1^2 v_2 w_{12} w_{21} x_1^2 x_2 \mu_1 \mu_2
+ v_1^2 v_2 w_{11} w_{22} x_1^2 x_2 \mu_1 \mu_2 \\
&\quad
- v_1 v_2^2 w_{12} w_{21} x_1 x_2^2 \mu_1 \mu_2
+ v_1 v_2^2 w_{11} w_{22} x_1 x_2^2 \mu_1 \mu_2 \\
&\quad
+ v_1^2 v_2^2 w_{12} w_{21} x_1^2 x_2^2 \mu_1 \mu_2
- v_1^2 v_2^2 w_{11} w_{22} x_1^2 x_2^2 \mu_1 \mu_2 \\
&\quad
+ v_1 v_3 w_{33} x_3 \mu_3
+ v_2 v_3 w_{33} x_3 \mu_3
- v_1 v_3^2 w_{33} x_3^2 \mu_3
- v_2 v_3^2 w_{33} x_3^2 \mu_3 \\
&\quad
+ v_1 v_3 w_{13} w_{31} x_1 x_3 \mu_1 \mu_3
- v_1 v_3 w_{11} w_{33} x_1 x_3 \mu_1 \mu_3 \\
&\quad
- v_1^2 v_3 w_{13} w_{31} x_1^2 x_3 \mu_1 \mu_3
+ v_1^2 v_3 w_{11} w_{33} x_1^2 x_3 \mu_1 \mu_3 \\
&\quad
- v_1 v_3^2 w_{13} w_{31} x_1 x_3^2 \mu_1 \mu_3
+ v_1 v_3^2 w_{11} w_{33} x_1 x_3^2 \mu_1 \mu_3 \\
&\quad
+ v_1^2 v_3^2 w_{13} w_{31} x_1^2 x_3^2 \mu_1 \mu_3
- v_1^2 v_3^2 w_{11} w_{33} x_1^2 x_3^2 \mu_1 \mu_3 \\
&\quad
+ v_2 v_3 w_{23} w_{32} x_2 x_3 \mu_2 \mu_3
- v_2 v_3 w_{22} w_{33} x_2 x_3 \mu_2 \mu_3 \\
&\quad
- v_2^2 v_3 w_{23} w_{32} x_2^2 x_3 \mu_2 \mu_3
+ v_2^2 v_3 w_{22} w_{33} x_2^2 x_3 \mu_2 \mu_3 \\
&\quad
- v_2 v_3^2 w_{23} w_{32} x_2 x_3^2 \mu_2 \mu_3
+ v_2 v_3^2 w_{22} w_{33} x_2 x_3^2 \mu_2 \mu_3 \\
&\quad
+ v_2^2 v_3^2 w_{23} w_{32} x_2^2 x_3^2 \mu_2 \mu_3
- v_2^2 v_3^2 w_{22} w_{33} x_2^2 x_3^2 \mu_2 \mu_3
\end{aligned}
\]

\[
\begin{aligned}
c_0=&-v_1 v_2 v_3
+ v_1 v_2 v_3 w_{11} x_1 \mu_1
- v_1^2 v_2 v_3 w_{11} x_1^2 \mu_1
+ v_1 v_2 v_3 w_{22} x_2 \mu_2
- v_1 v_2^2 v_3 w_{22} x_2^2 \mu_2 \\
&\quad
+ v_1 v_2 v_3 w_{12} w_{21} x_1 x_2 \mu_1 \mu_2
- v_1 v_2 v_3 w_{11} w_{22} x_1 x_2 \mu_1 \mu_2 \\
&\quad
- v_1^2 v_2 v_3 w_{12} w_{21} x_1^2 x_2 \mu_1 \mu_2
+ v_1^2 v_2 v_3 w_{11} w_{22} x_1^2 x_2 \mu_1 \mu_2 \\
&\quad
- v_1 v_2^2 v_3 w_{12} w_{21} x_1 x_2^2 \mu_1 \mu_2
+ v_1 v_2^2 v_3 w_{11} w_{22} x_1 x_2^2 \mu_1 \mu_2 \\
&\quad
+ v_1^2 v_2^2 v_3 w_{12} w_{21} x_1^2 x_2^2 \mu_1 \mu_2
- v_1^2 v_2^2 v_3 w_{11} w_{22} x_1^2 x_2^2 \mu_1 \mu_2 \\
&\quad
+ v_1 v_2 v_3 w_{33} x_3 \mu_3
- v_1 v_2 v_3^2 w_{33} x_3^2 \mu_3 \\
&\quad
+ v_1 v_2 v_3 w_{13} w_{31} x_1 x_3 \mu_1 \mu_3
- v_1 v_2 v_3 w_{11} w_{33} x_1 x_3 \mu_1 \mu_3 \\
&\quad
- v_1^2 v_2 v_3 w_{13} w_{31} x_1^2 x_3 \mu_1 \mu_3
+ v_1^2 v_2 v_3 w_{11} w_{33} x_1^2 x_3 \mu_1 \mu_3 \\
&\quad
- v_1 v_2 v_3^2 w_{13} w_{31} x_1 x_3^2 \mu_1 \mu_3
+ v_1 v_2 v_3^2 w_{11} w_{33} x_1 x_3^2 \mu_1 \mu_3 \\
&\quad
+ v_1^2 v_2 v_3^2 w_{13} w_{31} x_1^2 x_3^2 \mu_1 \mu_3
- v_1^2 v_2 v_3^2 w_{11} w_{33} x_1^2 x_3^2 \mu_1 \mu_3 \\
&\quad
+ v_1 v_2 v_3 w_{23} w_{32} x_2 x_3 \mu_2 \mu_3
- v_1 v_2 v_3 w_{22} w_{33} x_2 x_3 \mu_2 \mu_3 \\
&\quad
- v_1 v_2^2 v_3 w_{23} w_{32} x_2^2 x_3 \mu_2 \mu_3
+ v_1 v_2^2 v_3 w_{22} w_{33} x_2^2 x_3 \mu_2 \mu_3 \\
&\quad
- v_1 v_2 v_3^2 w_{23} w_{32} x_2 x_3^2 \mu_2 \mu_3
+ v_1 v_2 v_3^2 w_{22} w_{33} x_2 x_3^2 \mu_2 \mu_3 \\
&\quad
+ v_1 v_2^2 v_3^2 w_{23} w_{32} x_2^2 x_3^2 \mu_2 \mu_3
- v_1 v_2^2 v_3^2 w_{22} w_{33} x_2^2 x_3^2 \mu_2 \mu_3 \\
&\quad
- v_1 v_2 v_3 w_{13} w_{22} w_{31} x_1 x_2 x_3 \mu_1 \mu_2 \mu_3
+ v_1 v_2 v_3 w_{12} w_{23} w_{31} x_1 x_2 x_3 \mu_1 \mu_2 \mu_3 \\
&\quad
+ v_1 v_2 v_3 w_{13} w_{21} w_{32} x_1 x_2 x_3 \mu_1 \mu_2 \mu_3
- v_1 v_2 v_3 w_{11} w_{23} w_{32} x_1 x_2 x_3 \mu_1 \mu_2 \mu_3 \\
&\quad
- v_1 v_2 v_3 w_{12} w_{21} w_{33} x_1 x_2 x_3 \mu_1 \mu_2 \mu_3
+ v_1 v_2 v_3 w_{11} w_{22} w_{33} x_1 x_2 x_3 \mu_1 \mu_2 \mu_3
\end{aligned}
\]

\subsection{Cross-activation with one decay parameter}

Consider the general system (\ref{main6}) with the following set of parameters: $\mu_1=\mu_2=\mu_3=4,$ $w_{11}=w_{22}=w_{33}=1,$ $\theta_1=\theta_2=0.5, \; \theta_3=0.5/v_3,$ $v_1=1,$ $v_2=1,$ and $v_3$ is a free parameter.

The system is
\begin{equation} \label{maincrossactivation}
\left\{\begin{array}{rl}
&x'_1 =  \frac{1}{1+e^{-4( x_1  - 0.5)}}- x_1,    \\\\
&x'_2 = \frac{1}{1+e^{-4 (x_2 - 0.5)}} -  x_2, \\\\
&x'_3 = \frac{1}{1+e^{-4 (x_3 -
0.5)}}- v_3 x_3.
\end{array}\right.
\end{equation}
The critical points are solutions of
\begin{equation} \label{maincrossactivationcritpoints}
\left\{\begin{array}{rl}
&x_1 =  \frac{1}{1+e^{-4( x_1  - 0.5)}},    \\\\
&x_2 = \frac{1}{1+e^{-4 (x_2 - 0.5)}}, \\\\
&v_3 x_3 = \frac{1}{1+e^{-4 (x_3 -
0.5)}}.
\end{array}\right.
\end{equation}

Let $v_3=1.$ The nullclines are three orthogonal planes in the phase space, and they intersect at the central location $(0.5,0.5,0.5)$. If the coefficient $v_3$ is  varied, the  only critical point at $(0.5,0.5,1/(2 v_3) )$ changes its location. The type of this point is defined by zeros of the characteristic polynomial
\begin{equation} \label{main-char-pol1}
P_1(\lambda)=-\lambda^3 + (1-v_3) \lambda^2.
\end{equation}
Therefore, the characteristic numbers are $\lambda_1=\lambda_2=0,$ $\lambda_3=1-v_3.$ This means that the critical point $(0.5,0.5, 0.5/ v_3)$ is not stable for $0<v_3<1$ and is non-hyperbolic for $v_3>1.$ In the second case additional analysis is needed. 

\subsection{Transformation of a limit cycle}

Let the 3D system be
\begin{equation} \label{eq:3Dsystem}
\left\{
\begin{aligned}
\dot{x} &= 
\frac{1}{1 + e^{-4(k x -  z - \theta_1)}} 
- x, \\
\dot{y} &= 
\frac{1}{1 + e^{-4(-x + k y  - \theta_2)}} 
- y, \\
\dot{z} &= 
\frac{1}{1 + e^{-4 (-y + k z - \theta_3)}} 
- z,
\end{aligned}
\right.
\end{equation}
where $k$ is a parameter. For $k=1$ and the appropriate choice of $\theta_i$ ($\theta_i=0$ for $i=1,2,3$) this system has an attractive limit cycle depicted in Fig. \ref{fig:lc}.
\begin{figure}[h!] 
    \centering    \includegraphics[width=0.4\textwidth]{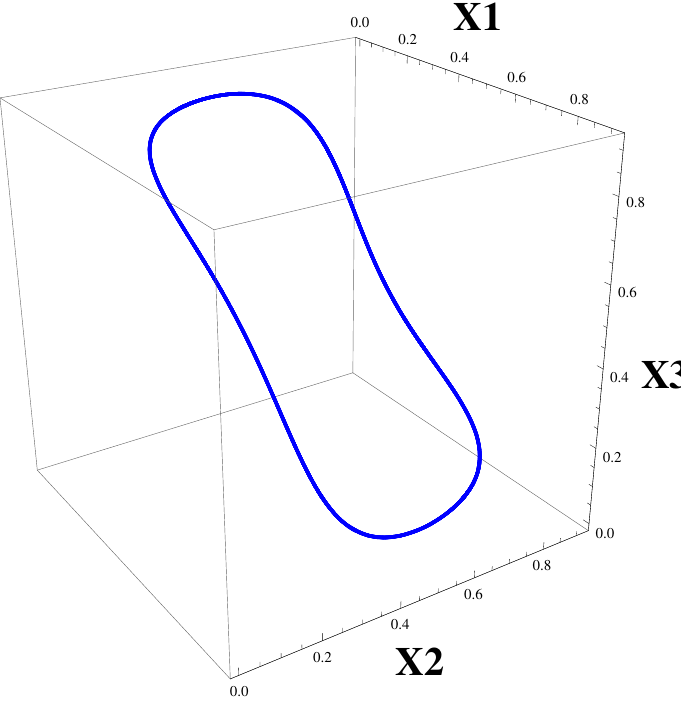}
    \caption{The limit cycle for $v_3=1.$}
    \label{fig:lc}
\end{figure}
What happens if the coefficient $v_3$ changes?
Let $v_3$ be a free parameter, $v_1=v_2=1$. We will change $v_3$ and find the characteristic numbers for the central critical point at $(0.5, 0.5, 0.5/v_3)$. To keep this critical point at the central location we have to set $\theta_1=0.5-0.5/v_3, \; \theta_2=0, \; \theta_3=-0.5+0.5/v_3.$ 

The coefficient matrix $A2$ of the linearized at the central critical point $(0.5, 0.5, 0.5/v_3)$ system is
\begin{equation} \label{A2}
A2=\left(
\begin{matrix}
-1
    & 0 
    & -1 \\[6pt]
-1
    & 0
    & 0 \\[6pt]
0
    & -1
    & -1 - v_3
\end{matrix}
\right).
\end{equation}
The characteristic polynomial $\det[A2-\lambda E]$ is
\begin{equation} \label{char-pol2}
P(\lambda)=-\lambda^3+(1-v_3) \lambda^2 -1
\end{equation}

The characteristics of a critical point are saved in Table \ref{table2}.

\begin{table}[h!] 
\centering
\caption{}
\label{table2}

\begin{tabular}{|c|c|c|}
\hline
$v_3$ & Point & Characteristic numbers ($\lambda_1, \lambda_2, \lambda_3$) \\ \hline
$0.6$    & $0.5, 0.5, 0.833$  & $(-0.882, 0.641 \pm 0.849 \mathrm{i})$  \\ \hline
$1.0$    & ${0.5, 0.5, 0.5}$    & $(-1, 0.5\pm 0.866 \mathrm{i})$  \\ \hline
$1.5$     & ${0.5, 0.5, 0.333}$    & $(-1.197, 0.348 \pm 0.844\mathrm{i})$  \\ \hline
$2.5$     & ${0.5, 0.5, 0.2}$    & $(-1.806, 0.153 \pm 0.728 \mathrm{i})$    \\ \hline
$5.5$    & ${0.5, 0.5, 0.091}$    & $(-4.548, 0.024 \pm 0.468 \mathrm{i})$    \\ \hline
$9.5$    & ${0.5, 0.5, 0.052}$    & $(-8.513, 0.006 \pm 0.342 \mathrm{i})$   \\ \hline
\end{tabular}

\vspace{0.5em}
{\normalsize Characteristics of the central critical point.}
\end{table}

The central critical points mentioned in Table \ref{table2} are not attractive. Despite of the dominance of the real root over the real parts of complex roots chaotic behavior of solutions was not observed (see subsection 3.5).

\subsection{All $v_i$ are free parameters}

Let the regulatory matrix be as before
\begin{equation} \label{W3}
W3=\left(
\begin{matrix}
0
    & 0 
    & -1 \\[6pt]
-1
    & 0
    & 0 \\[6pt]
0
    & -1
    & 1
\end{matrix}
\right).
\end{equation}
Suppose $\mu_i=4$, $i=1,2,3;$\\
\begin{equation} \label{thetas}
\begin{aligned}
\theta_1 &= w_{11}\frac{0.5}{v_1} + w_{12}\frac{0.5}{v_2} + w_{13}\frac{0.5}{v_3},\\
\theta_2 &= w_{21}\frac{0.5}{v_1} + w_{22}\frac{0.5}{v_2} + w_{23}\frac{0.5}{v_3},\\
\theta_3 &= w_{31}\frac{0.5}{v_1} + w_{32}\frac{0.5}{v_2} + w_{33}\frac{0.5}{v_3}.
\end{aligned}
\end{equation}
where $v_i$ are now free parameters. Such choice of $\theta_i$ will ensure that the central point $(0.5/v_1, 0.5/v_2, 0.5/v_3)$ of the invariant set $Q_3$ is the critical point (Proposition 1). We wish to know how the choice of $v_i$ will affect the type of this point. For this, construct the characteristic polynomial, which depends on all decay parameters $v_i.$

\begin{equation} \label{A3}
A3=\left(
\begin{matrix}
-v_1 
    & 0 
    & -1 \\[6pt]
-1
    & -v_2
    & 0 \\[6pt]
0
    & -1
    & 1 - v_3
\end{matrix}
\right).
\end{equation}
The characteristic equation $\det[A3 - \lambda E]=0$ is 
\begin{equation} \label{charpol3}
- \lambda^3 +P \lambda^2 + Q \lambda + R=0, 
\end{equation}
where\\
$P=1- v_1 - v_2 - 
    v_3$\\
    $Q= v_1 + v_2 - v_1 v_2 - v_1 v_3 - v_2 v_3$\\
    $R=-1 + v_1 v_2 - v_1 v_2 v_3.$\\
For some positive triples $(v_1, v_2, v_3)$ the characteristic equation for the central critical point has real roots of different signs. The corresponding point is non-attractive and one can expect the existence of attractor of more complicated nature. Consider the example.  

\vspace{1em}
{\bf Example.} Set $(v_1,v_2, v_3)$ to $(1.7, 2, 0.7).$ The characteristic equation is
\begin{equation} \label{charpol4}
-\lambda^3 - 3.40\lambda^2 - 
  2.29 \lambda + 0.02=0
\end{equation}
with roots $(\lambda_1=-2.469 \; \lambda_2=-0.939, \; \lambda_3=0.008).$ The central critical point in the region $Q_3=\{0<x<0.58, \; 0<y<0.5,$ $ \; 0<z<1.43\}$ is a 3D saddle.

\begin{figure}[h!] 
    \centering    \includegraphics[width=0.35\textwidth]{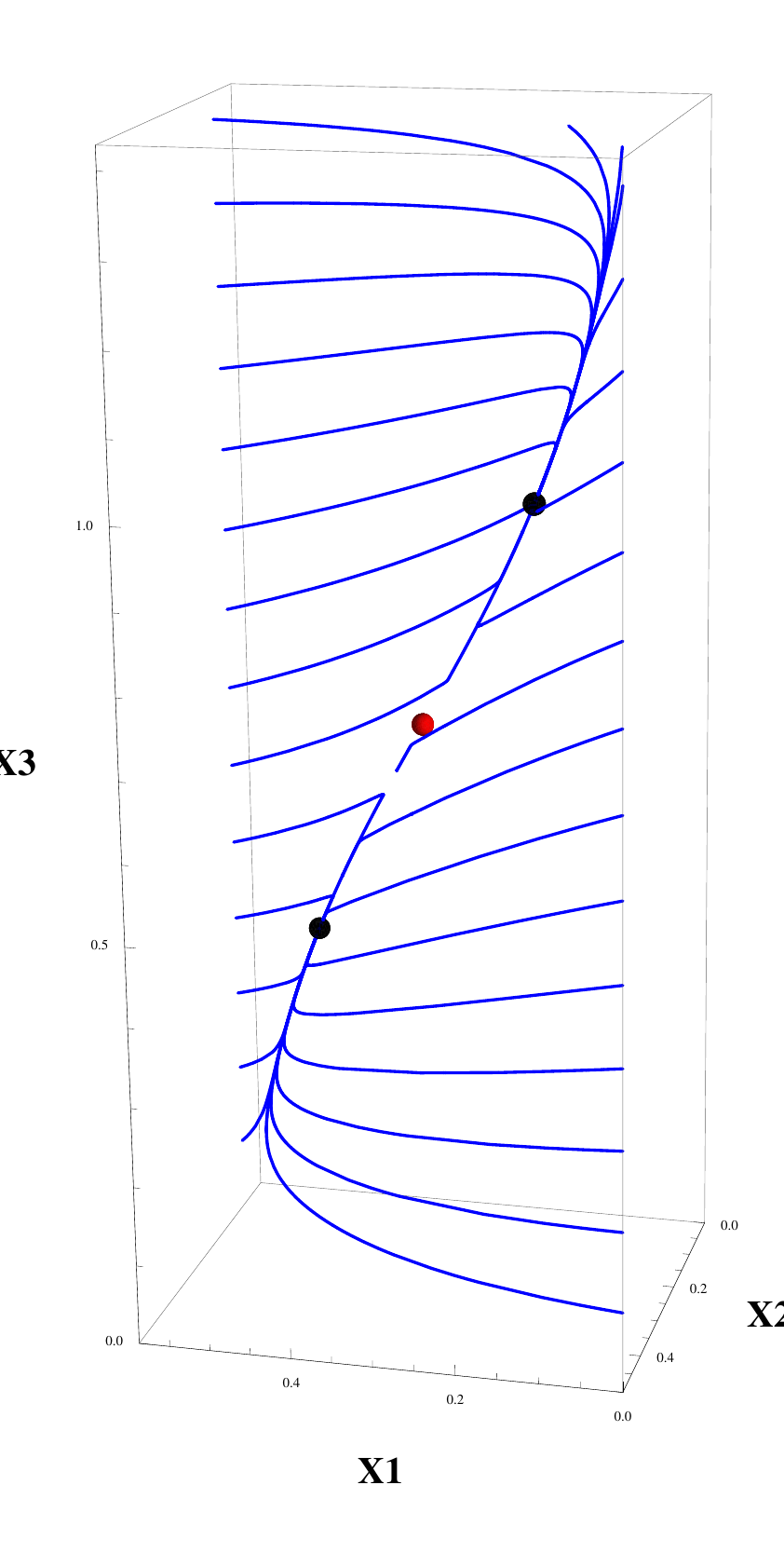}
    \caption{The 30 trajectories emanating from $(0.58, 0.1, 0.1 + k)$ and $(0.0, 0.5, 0.1 + k)$, $k=0, 0.1, \dots, 1.4.$ The central point is marked red. Two symmetric critical points that attracts trajectoriesbare marked in orange.}
    \label{fig:30}
\end{figure}

There are two symmetric attractive points at $(0.4486 0.1751, 0.4223)$ and $0.1395, 0.3249, 1.006).$ They are marked in Figure \ref{fig:30}.

\subsection{Influence on the number of equilibria}

Change of coefficients $v_i$ may influence the number of critical points (equilibria). The set of critical points for GRN type systems is never empty.
Their number may vary from one and up, but is finite. The number of critical points for the two-dimensional system can be at most nine. The two nullclines (recall that critical points are the intersection points of nullclines) then are two Z-shaped curves.

Consider a simple example of how $v_i$ coefficients affect the number of equilibria. Let $\mu_i=\mu=4,$ $\theta_i=1$ for $i=1,2,3.$ Consider the case of uncoupled system, where the regulatory matrix is
\begin{equation} \label{A4}
A4=\left(
\begin{matrix}
1
    & 0 
    & 0 \\[6pt]
0
    & 1
    & 0 \\[6pt]
0
    & 0
    & 1
\end{matrix}
\right).
\end{equation}
Suppose $v_1=v_2=v_3=v.$ So all parameters are set, except $v,$ which takes values $0.1, 0.3, 0.5, 0.7.$ 
The critical points $(z,z,z)$ satisfy the relation
\[
v z= \frac{1}{1+exp\{-4 (z-1)\}}.
\]
The explanatory Figure \ref{fig:vis}
\begin{figure}[h!] 
    \centering    \includegraphics[width=0.4\textwidth]{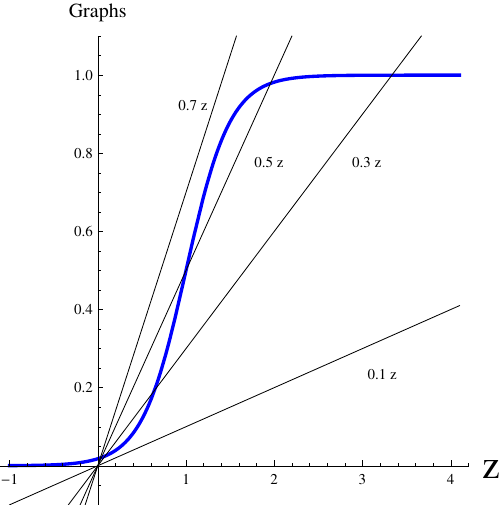}
    \caption{Several solutions of the system that avoid going to the central critical point (red ball).}
    \label{fig:vis}
\end{figure}
illustrates four cases, when there are three equilibria vor $v=0.3$ and $v=0.5,$ one equilibrium for $v=0.7$ (close to the origin), and one equilibrium for $v=0.1$ (not seen in the plot because the corresponding straight line meets the sigmoid at $z \approx 10$.)

 \subsection{Particular cases and Lyapunov curves tests}

 \subsubsection{Existing example of chaos.} Consider example taken from \cite{Kozlovska}. Let the matrix $W$ be
\begin{equation} \label{W4}
W4=\left(
\begin{matrix}
0
    & 0 
    & 2 \\[6pt]
-0.82
    & -0.2
    & 4.55 \\[6pt]
0.1
    & -0.87
    & 1.11
\end{matrix}
\right).
\end{equation}
The set of the remaining  parameters:\\
$v_1=0.164, \: v_2=0.1, \: v_3=0.2,$\\
$\mu_1=4.38, \: \mu_2=4, \: \mu_3=3.1,$\\
$\theta_1=0.77, \: \theta_2=1.09, \: \theta_3=0.62.$\\
For this set of parameters, the critical point was detected at $P=(6.096, 2.085, 1.354).$ Chaotic behavior of solutions was observed. For the initial conditions $(x_0, y_0, z_0)=(0.41, 0.1, 0.4)$ the test on Lyapunov exponents was performed. The result is depicted in Figure \ref{fig:lcurves6}, where the Lyapunov curves and the Kaplan-Yorke degree indicate the tendency to chaotic behavior of solutions. 
\begin{figure}[h!] 
    \centering    \includegraphics[width=0.4\textwidth]{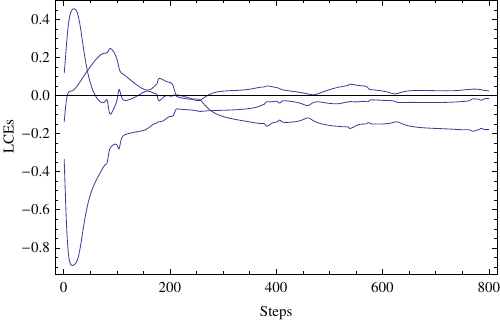}
    \caption{Lyapunov curves for the above set of parameters.}
    \label{fig:lcurves6}
\end{figure}

Let us change only $v_2$ slightly. Let $v_2=0.3.$ 

\begin{figure}[h!] 
    \centering    \includegraphics[width=0.4\textwidth]{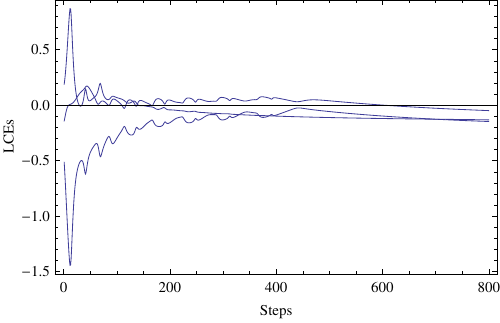}
    \caption{Lyapunov curves for the above set of parameters, where $v_2=0.3$.}
    \label{fig:lcurves5a}
\end{figure}
Lyapunov curves tend to go down (Figure \ref{fig:lcurves5a}), leaving the positive region. Lyapunov exponents are all negative $(-0.0494882, -0.132788, -0.147422)$, Kaplan-Yorke topological dimension is not provided.

\vskip0.3cm
\subsubsection{Changing one parameter only.} Consider an example of the system described by  the data (\ref{W3}), (\ref{thetas}), $\mu_i=4.$  The matrix of coefficients of the linearized system around the central critical point $(1/(2 v_1),1/(2 v_2),1/(2 v_3))$ is

\begin{equation} \label{A5}
A5=\left(
\begin{matrix}
-v_1
    & 0 
    & -1 \\[6pt]
-1
    & -v_2
    & 0 \\[6pt]
0
    & -1
    & 1-v_3
\end{matrix}
\right).
\end{equation}

We will change the decay parameters $(v_1, v_2, v_3)$ and show the characteristic numbers for the corresponding central critical points and Lyapunov curves. Making test on Lyapunov curves we assume that the initial condition for the testing trajectory is $(0.4/v_1, 0.1/v_2, 0.4/v_3),$ so the initial conditions change proportionally with the parameter.

Case 1. Let $(v_1,v_2,v_3)$ be $(0.5,0.5,0.5).$ The characteristic numbers are $\lambda_1=-1.254,$ $\lambda_{2,3}=0.377 \pm 0.744 \,\mathbf{i}.$ The Lyapunov curves are shown in Figure \ref{fig:lcurves8a}
\begin{figure}[h!] 
    \centering    \includegraphics[width=0.4\textwidth]{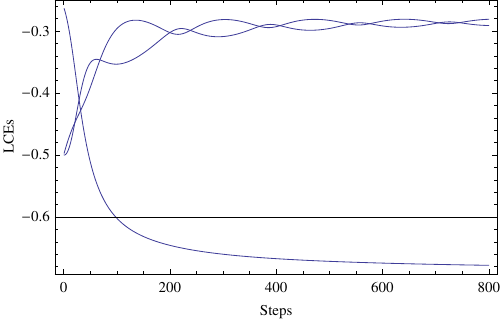}
    \caption{Lyapunov curves for Case 1.}
    \label{fig:lcurves8a}
\end{figure}

Case 2. Let $(v_1,v_2,v_3)$ be $(1, 1, 0.6).$ The characteristic numbers are $\lambda_1=-1.691,$ $\lambda_{2,3}=0.045 \pm 0.593 \,\mathbf{i}.$ The Lyapunov curves are shown in Figure \ref{fig:lcurves9a}
\begin{figure}[h!] 
    \centering    \includegraphics[width=0.4\textwidth]{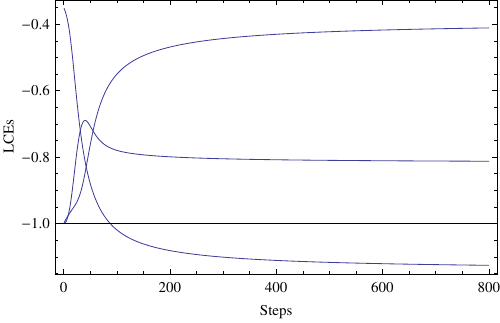}
    \caption{Lyapunov curves for Case 2.}
    \label{fig:lcurves9a}
\end{figure}

Case 3. Let $(v_1,v_2,v_3)$ be $(1, 1, 1.5).$ The characteristic numbers are $\lambda_1=-1.858,$ $\lambda_{2,3}=0.320 \pm 0.839 \,\mathbf{i}.$ The Lyapunov curves are shown in Figure \ref{fig:lcurves9a}
\begin{figure}[h!] 
    \centering    \includegraphics[width=0.4\textwidth]{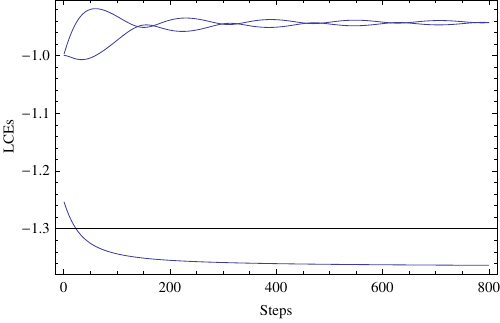}
    \caption{Lyapunov curves for Case 3.}
    \label{fig:lcurves10a}
\end{figure}

Case 4. Let $(v_1,v_2,v_3)$ be $(1, 1, 0.1).$ The characteristic numbers are $\lambda_1=-1.628,$ $\lambda_2=0.172,$ $\lambda_3=0.356.$ The Lyapunov curves are shown in Figure \ref{fig:lcurves11a}
\begin{figure}[h!] 
    \centering    \includegraphics[width=0.4\textwidth]{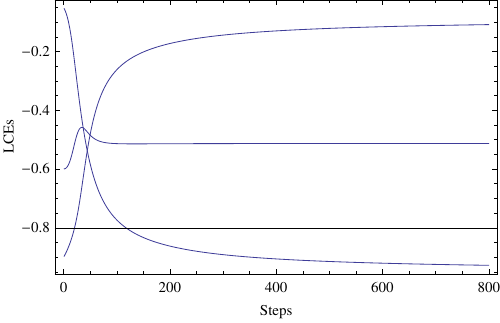}
    \caption{Lyapunov curves for Case 4.}
    \label{fig:lcurves11a}
\end{figure}
 
Changing $v_3$ only and keeping $v_1=v_2=1,$ one observes: if $v_3$ is near zero, there are three real characteristic numbers of different signs. The central critical point is non-attractive. When increasing $v_3$ to the value $0.6$ thereal characteristic number is negative, and two complex conjugate roots have positive real part. The central critical point is non-attractive. Further increasing $v_3$ up to the value $1.5$ leads the substantial change. The central critical point becomes fully attractive; both the real root and the real part of the complex roots are negative. Both cases 2 and 3 yield non-attraction at the center. So another attractor should exist. Both attractors are visualized in the below two figures.

\begin{figure}[h!] \label{last}
\centering
\begin{subfigure}{0.37\textwidth}
    \centering
    \includegraphics[width=\linewidth]{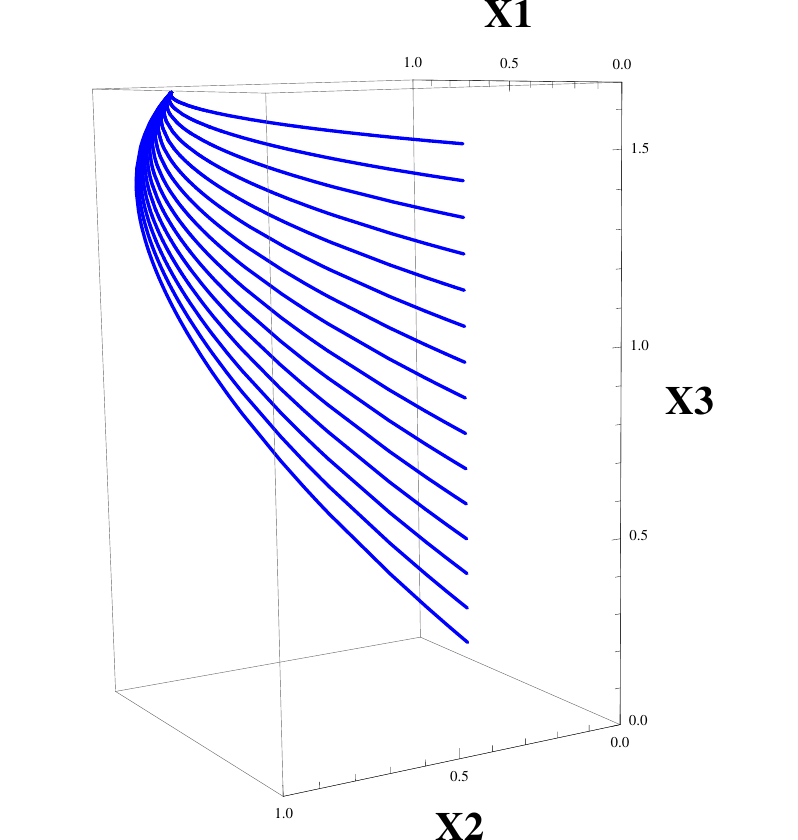}
    \caption{Case 2. Trajectories go the attractor at $(0.504, 1, 1.662).$}
\end{subfigure}
\hfill
\begin{subfigure}{0.35\textwidth}
    \centering
    \includegraphics[width=\linewidth]{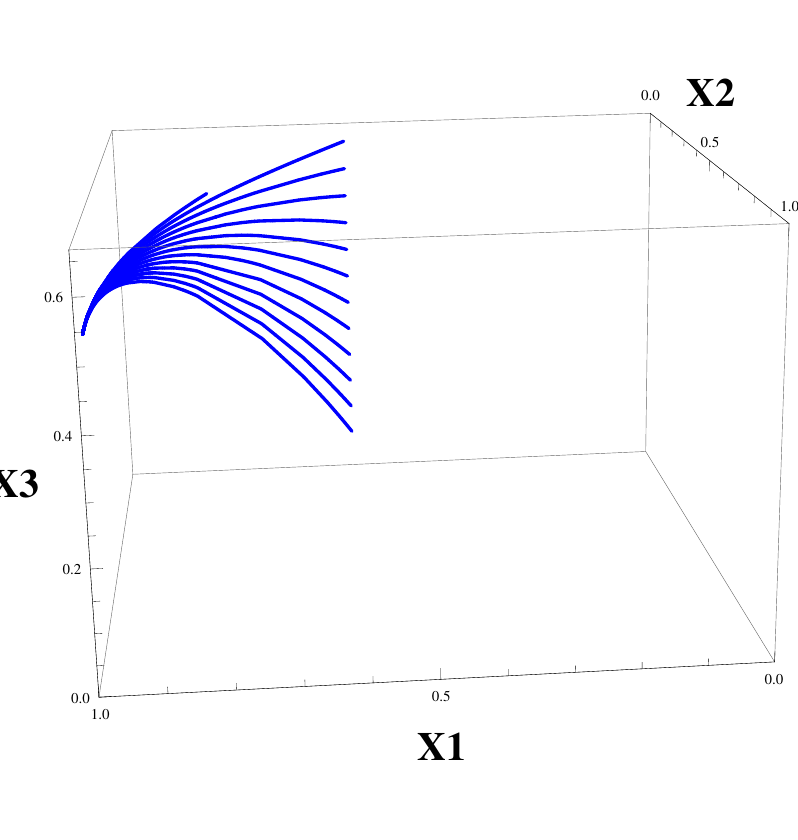}
    \caption{Case 3. Trajectories go to the attractor (critical point) at $(0.99, 1, 0.54).$}
\end{subfigure}
\end{figure}

In calculations of this subsection the program by M. Sandri \cite{Sandri} was used.

\section{Conclusions}
We have considered GRN type systems that are of great importance in mathematical modeling of real genetic networks. These systems allow tracing the evolution of solutions and predicting the behavior of elements of a network over finite intervals. This system consists of two parts: the nonlinear part, represented by a sigmoidal function, and the linear part. Dominance of the linear part simplifies the behavior of solutions. In the absence of nonlinearity, all solutions exponentially go to zero. When balance between both parts of a system is established, the phase space becomes more complicated, revealing nonlinear effects, such as coexistence of different attractors and a sequence of bifurcations under the change of parameters. We focused on coefficients in the linear part. Their role is substantial. The analysis made in Section 2 demonstrates the technique of working with three-dimensional systems, up to constructing the system with desired properties. In further sections, we discuss critical points, their interaction with decay coefficients, and locating a critical point at the center of the invariant set. Taking as a basis the 3D system that has an attractor in the form of a stable limit cycle, we show how the behavior of solutions may change under the change of decay coefficients. A set of formulas for making the local analysis around the critical point was obtained. Using these formulas to study some particular cases, we show how properties of a critical point can be affected.
The number of critical points itself can be changed by the appropriate change of decay coefficients. The chaotic behavior is not easy to detect in GRN systems. We examined an example in the literature \cite{Kozlovska} showing how coefficients $v_i$ influence that. Examples in Section 2, where dominance of the real characteristic numbers is established over the real parts of the remaining complex roots, can be further studied to find more chaotic attractors.

\end{document}